\documentclass[12pt]{article}
\usepackage{a4wide}

\usepackage{amsmath}
\usepackage{amssymb}
\usepackage{amsthm}
\usepackage{charter}
\usepackage{color}
\usepackage{graphicx}
\usepackage[normalem]{ulem}
\usepackage[ruled, vlined, linesnumbered]{algorithm2e}

\graphicspath{{figs/}}

\newcommand{\drop}[1]{}
\newcommand{\ncgmf}[0]{\textsc{nc-gmf}}
\newcommand{\gmf}[0]{\textsc{gmf}}
\newcommand{\ftp}[0]{\textsc{ftp}}

\theoremstyle{definition}
\newtheorem{theorem}{Theorem}

\newtheorem{definition}[theorem]{Definition}

\newtheorem{property}[theorem]{Property}

\newcommand{\equals}{\stackrel{\mathrm{def}}{=}}

\newcommand{\RBF}{\textsc{rbf}}
\newcommand{\MRBF}{\textsc{mrbf}}
\newcommand{\LB}{\textsc{lb}}
\newcommand{\UB}{\textsc{ub}}

\newcommand{\task}{\tau}

\title{Sufficient FTP Schedulability Test for the Non-Cyclic Generalized Multiframe Task Model}
\author{
	Vandy Berten\\ Universit\'e Libre de Bruxelles \\ Brussels, Belgium \and
	Jo\"el Goossens \\ Universit\'e Libre de Bruxelles \\ Brussels, Belgium
	}

\begin{document}
\maketitle

\section{Introduction}
In this paper we consider the schedulability of the Non-Cyclic Generalized Multiframe Task Model System Model defined by Moyo et al.~\cite{moyo2010schedulability} (\ncgmf{} in short). Moyo et al.\@ provided a sufficient EDF schedulability test, Baruah extended the model by considering the Non-cyclic recurring task model~\cite{baruah2010non} and provided an exact feasibility test. Stigge et al.\@ further extend this model and consider Digraph real-time task model~\cite{digraph}, and consider the feasibility (or, equivalently, the EDF-schedulability) of their task model.

Formally, an \ncgmf{} task set $\task$ is composed of $M$ tasks $\{\task_1, \dots, \task_M\}$. An \ncgmf{} task $\task_i$ is characterized by the 3-tuple $(\overrightarrow{C_i}, \overrightarrow{D_i}, \overrightarrow{T_i})$ where $\overrightarrow{C_i}$, $\overrightarrow{D_i}$, and $\overrightarrow{T_i}$ are $N$-ary vectors $[C_i^1, C_i^2, \ldots, C_i^{N_i}]$  of execution requirement, $[D_i^1, D_i^2,\ldots, D_i^{N_i}]$ of (relative) deadlines, and $[T_i^1, T_i^2, \ldots,T_i^{N_i}]$ of minimum separations, respectively. Additionally each task makes an initial release not before time-instant $O_i$. 

The interpretation is as follows: each task $\tau_i$ can be in $N_i$ different configurations, where configuration $j$ with $j \in \{1, \dots, N_i\}$ corresponds to the triplet $(C_i^j, D_i^j, T_i^j)$. We denote by $\ell_i^k \in \{1, \dots, N_i\}$ the configuration of the $k^\text{th}$ frame of task $\tau_i$. The first frame of task $\task_i$ has an arrival time $a_i^1 \geq O_i$, an execution requirement of $C_i^{\ell_i^1}$, an absolute deadline $a_i^1+D_i^{\ell_i^1}$ and a minimum separation duration of $T_i^{\ell_i^1}$ with the next frame. 
The $k^\text{th}$ frame has an arrival time $a_i^k \geq a_i^{k-1}+T_i^{\ell_i^{k-1}}$, an execution requirement of $E_i^{\ell_i^k}$, an absolute deadline $a_i^k+D_i^{\ell_i^k}$ and a minimum separation duration of $T_i^{\ell_i^k}$ with the next frame.

It is important to notice that $\ell_i^k$ can be any value in $\{1, \dots, N_i\}$ for $i\in \{1, \dots, M\}$ and $k \in \mathbb{N}^+_0$. I.e., the exact order in which the different kinds of jobs are generated is not know at design-time, it is \emph{non-cyclic}. In this way this model generalizes the \gmf{} model~\cite{Baruah1999Generalized-mul}.

In this research we consider the scheduling of \ncgmf{} using Fixed Task Priority (\ftp) schedulers and without loss generality $i<j$ implies that $\task_{i}$ has a higher priority than $\task_{j}$. Notice that, while having a model and solutions close to ours, \cite{digraph} considers dynamic priority tasks.
Even if a very similar mechanism (the request bound function) works very well with dynamic priorities allowing one to provide necessary and sufficient feasibility tests, it introduces some pessimism with fixed priorities, leading to a sufficient schedulability test.

\noindent\textbf{This research.} Our goal is to provide a sufficient schedulability test ---ideally polynomial--- for the scheduling of \ncgmf{} using \ftp{} schedulers. We report two first results: (i) we present and prove correct the critical instant for the \ncgmf{} then (ii) we propose an algorithm which provides a sufficient (but pseudo-polynomial) schedulability test. 

\section{Critical Instant}

First notice that the popular critical instant notion (see Definition~\ref{def:crit}) is not really applicable in our task model.

\begin{definition}\label{def:crit} A \emph{critical instant} of a task $T_{i}$ is a time instant such that: the job of $T_{i}$ at this instant has the maximum response time of all job of $T_{i}$.
\end{definition} 

This notion is difficult to extend to non-cyclic model in the sense that the workload (task requests) is not unique, contrary to the Liu and Layland or the (cyclic) \gmf{} task model.

Anyway, the critical instant is important because it provides the \emph{worst-case scenario} from $\tau_i$ standpoint, in the next section we will show that the synchronous case remains the worst case for the more general model namely the \ncgmf.

\begin{definition} Let us consider an \ncgmf{} task set $\{\task_i : i \in [1, \dots, M]\}$, where $\task_i$ is described by  
$$\{(C_i^1, D_i^1, T_i^1), \dots, (C_i^{N_i}, D_i^{N_i}, T_i^{N_i})\}.$$

The \emph{scenario of task $\task_i$}, denoted by $\mathcal{S}_i$, is a value $a_i^1\ge O_i$ (the initial release time) associated with an infinite sequence of tuples  $\{(\ell_i^k, \pi_i^k): k=1, 2\dots\}$, with:
\begin{itemize}
\item $\ell_i^k \in \{1, \dots, N_i\}$,
\item $\pi_i^k \ge T_i^{\ell_i^k}$.
\end{itemize}

If task $\task_i$ follows the scenario $\mathcal{S}_i$, then:
\begin{itemize}
\item the first job is released at time $a_i^1$, has an execution requirement of $C_i^{\ell_i^1}$ and a relative  deadline $D_i^{\ell_i^1}$,
\item the $k^\text{th}$ job ($k>1$) is released at time $a_i^k = a_i^{k-1} + \pi_i^{k-1}$,  has an execution requirement of $C_i^{\ell_i^k}$ and a relative deadline $D_i^{\ell_i^k}$.
\end{itemize}
\end{definition}


In this work we consider offset independent scenarios with the following definition:
\begin{definition} A scenario $\mathcal{S}_i$ is said to be \emph{offset independent} if it depends neither on $a_i^1$, nor on the schedule. In other words, if we shift by $\Delta$ the first arrival, all arrivals will be shifted the same way, and the jobs will keep the same characteristics.

A scenario $\widehat{\mathcal{S}}_i$ obtained in shifting $\mathcal{S}_i$ by $\Delta$ is  described by:
\begin{eqnarray*}
\widehat{a}_i^1 &=& a_i^1 - \Delta \\
\widehat{\ell}_i^k &=& \ell_i^k ~~\forall k \\
\widehat{\pi}_i^k &=& \pi_i^k ~~\forall k 
\end{eqnarray*}
\end{definition}

\begin{definition}[From \cite{RTCSA00}]
An \emph{idle instant} occurs at time $t$ in a schedule if all jobs arriving strictly before $t$ have completed their execution before or at time $t$ in the schedule.
\end{definition}

\begin{theorem}\label{thm:worstcase}
For any task $\tau_i$ the maximum response time occurs when all higher priority tasks release a new task request synchronously with $\tau_i$ and the subsequent requests arrive as soon as possible.
\end{theorem}
\begin{figure}
\centering \def\svgwidth{0.6\linewidth} 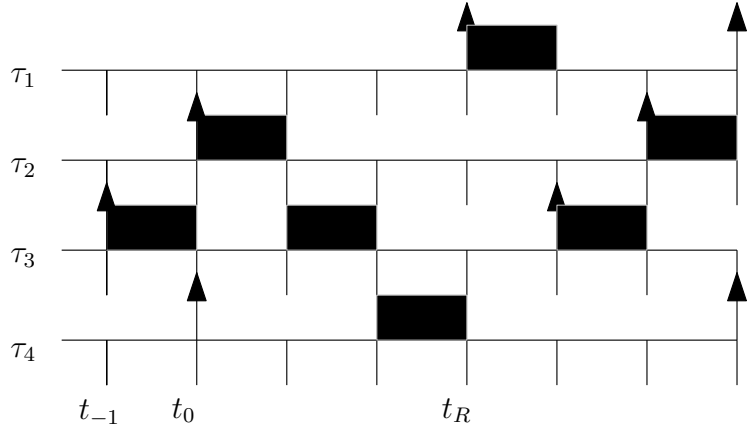
\caption{\label{fig:schedule} $t_{-1}, t_{0}$ and $t_{R}$.} 
\end{figure}

\begin{proof}
We\footnote{This proof is based upon lecture slides, the authors thank S.~Funk and J.~Anderson. It corrects~\cite{LiuLay73,Liu2000Real-Time-Syste} for the popular sporadic task model.} consider a job $J_{i}$ in a given scenario and we will show that, by shifting the job  releases of higher priority tasks, in the synchronous case the response time of the job $J_{i}$ does not decrease. Since that property holds whatever the scenario we can conclude that the worst case response time occurs in a synchronous scenario, hence the property.

Let $t_{0}$ be the $J_{i}$ release time. Let $t_{-1}$ be the latest idle instant for the scheduling of $\task_{1}, \ldots, \task_{i}$ at or before $t_{0}$. Let $t_{R}$ denote the time instant when $J_{i}$ completes. (See Figure~\ref{fig:schedule}.)

First, if we (artificially) redefine $J_{i}$ release time to be $t_{-1}$ then $t_{R}$ remains unchanged but the $J_{i}$ response time may increase.

Assume now that $\task_{i}$ releases a job at $t_{-1}$ and some higher priority task $\task_{j}$ does not. If we left-shift $\task_{j}$ such that its first job is released at time $t_{-1}$ then $J_{i}$ response time remains unchanged or is increased (since we consider offset independent scenarios). We have constructed a portion  of the schedule that is identical to the one which occurs at time $t_{-1}$ when $\task_{1}, \ldots, \task_{i}$ release job together. This shows that asynchronous releases cannot cause larger response time than synchronous release (for the same scenario).

It remains to show that the left-shift does not impact on $t_{-1}$. It may be noticed that no job shift pasts $t_{-1}$, every job that starts after $t_{-1}$ continues to do so after the shift by construction. Similarly, any job that starts before $t_{-1}$ continues to do so after the shift and $t_{-1}$ remain an idle instant: for each interval $[x,t_{-1}]$, no job shift into this interval, no left shifts cross $t_{-1}$, therefore the total demand during $[x, t_{-1}]$ cannot increase, therefore the last completion time prior to $t_{-1}$ cannot be delayed. 

The same argument (i.e., the left-shift) can be used to complete the proof and show that the worst case response time occurs when the subsequent requests arrive as soon as possible.
\end{proof}

\section{Schedulability Test}

\subsection{Request Bound Function}

\begin{figure*}
\centering \def\svgwidth{\linewidth} 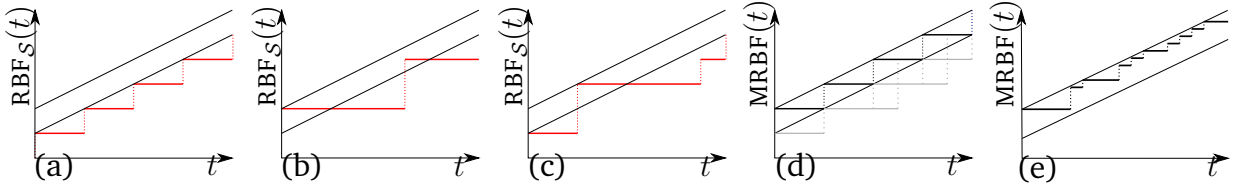
\caption{\label{fig:requestfunctions} (a, b, c): Request bound functions (\RBF) for a task $\task_i$  with $\protect\overrightarrow{C_i} = (1,2)$  and $\protect\overrightarrow{T_i} = (2, 5)$. In (a): $\mathcal{S} = (1, 1, 1, 1, ...)$, in (b):  $\mathcal{S} = (2, 2, ...)$, in (c): $\mathcal{S} = (1, 2, 1, 1, ...)$. (d): Maximum request bound functions (\MRBF) (all the request bound functions are in light gray). (e): Maximum request bound function for task $\tau_i$  with $\protect\overrightarrow{C_i} = (4, 6, 3)$  and $\protect\overrightarrow{T_i} = (5, 10, 4)$. The two oblique lines are $\UB_i$ and $\LB_i$ (using Equation~\ref{eq:LB}).}
\end{figure*}

In order to provide a schedulability test, we will use a similar approach to the one used by \cite{Tindell, RBF, digraph}, in order to bound the computation demand that a task can request. From the previous section, we know that a scenario has its worst demand when the first request arrives at time $0$, and all other requests arrive as soon as possible. I.e., $a_i^1 = 0$ and $\pi_i^k = T_i^{\ell_i^k}$. This kind of scenario will be called a \emph{dense scenario}. 

In the following, we define a \emph{sporadic scenario} as a scenario were exact release times are let free. A sporadic scenario is then determined by a sequence of configuration numbers $\{\ell_i^1, \ell_i^2, \dots\}$ with $\ell_i^k \in \{1, \dots, N_i\} ~\forall i, k$. Theorem~\ref{thm:worstcase} says then amongst all the possible arrival pattern of a sporadic scenario, the dense scenario is the worst.

From this, we provide a new definition:
\begin{definition}
Given a sporadic scenario $\mathcal{S}_i$ of a task $\task_i$, we define the \emph{Request Bound Function} $\RBF_{\mathcal{S}_i}(t)$ as the maximal total demand that $\tau_i$ can request in the interval $[0, t]$. Formally, 
$$
\RBF_{\mathcal{S}_i}(t) = \sum_{k=1}^{\alpha(t)} C_i^{\ell_i^k} \text{~with~} \alpha(t) = \min\{\beta\geq 0 \mid \sum_{k=1}^\beta T_i^{\ell_i^k} \geq t\}.
$$
\end{definition}

In the single frame model, where there is only one possible (sporadic) scenario, this corresponds to the Request Bound Function~\cite{RBF}. In the framework of feasibility for dynamic priorities, authors of \cite{digraph} consider Demand Bound Function, instead of Request Bound Function.

Some example of such request bound functions are shown in Figure~\ref{fig:requestfunctions}(a)--\ref{fig:requestfunctions}(c). We can observe that a request bound function is always composed of a sequence of connected steps, where the $k^\text{th}$ step is composed of a vertical jump of $C_i^{\ell_i^k}$, followed by a horizontal segment of length $T_i^{\ell_i^k}$.

Let us now assume that the (sporadic) scenario is fixed for all tasks $\task_1, \dots, \task_{i-1}$, is proved to be schedulable, and we want to know if the system is still schedulable if we add a task $\task_i$. Adapting the results from \cite{Tindell}, we need to check that, if starting synchronously with all higher priority tasks, $\task_i$ can still meet its deadline. As $\task_i$ can have different configurations, we need to find the smallest solution to the following $N_i$ inequalities:  
\begin{equation}\label{eq:RBF}
t = C_i^k + \sum_{j<i} \RBF_{\mathcal{S}_j}(t)~~\forall k \in \{1, \dots, N_i\}.
\end{equation}
If we now want to check the schedulability of any task system, we need to raise the hypothesis that the scenarios are known, and consider all the possible combinations of scenarios. 
But checking the schedulability through this method would be impossible within a reasonable amount of time, even for small problem instances.


In the following, we propose to bound all the possible request bound functions for a given task, and to use this bound instead of all possible scenarios. This will make the problem tractable, but at the price of loosing the exactness of the schedulability test (as explained in Section~\ref{sect:pessimism}). 

\subsection{Maximum Request Bound Function}
In the following, we will focus on the maximal request that a task can produce in the interval $[0, t)$, whatever the scenario.

\begin{definition}
The \emph{Maximum Request Bound Function} $\MRBF_i(t)$ of a task $\task_i$ represents an upper bound on the demand that task $\task_i$ can request up to any time $t$. Formally,
$$\MRBF_i(t) = \max_{\sigma \in \{\mathcal{S}_i\}} \RBF_\sigma(t)
$$
where $\{\mathcal{S}_i\}$ represents all the possible dense scenarios of $\task_i$.
\end{definition}

Notice the \MRBF{} is not the maximum of all the \RBF's (such a maximum does not exist in general), but a function which, for each $t$, gives the maximum of all $\RBF(t)$.

Examples of such \MRBF{} functions can be seen in Figures~\ref{fig:requestfunctions}(d)--\ref{fig:requestfunctions}(e). 
Before explaining how this $\MRBF$ can be practically computed, here is how we use it:

\begin{theorem}
An \ncgmf{} task set $\task$ is schedulable using an \ftp{} assignment if, $\forall i \in \{1, \dots, M\}, \forall k \in \{1, \dots, N_i\}$, the smallest positive solution to the equality
$$
t = C_i^k + \sum_{j<i} \MRBF_j(t)
$$
is not larger than $D_{i}^{k}$.
\end{theorem}

\begin{proof}
This is a consequence of Theorem~\ref{thm:worstcase} and the worst-case response time computation for \ftp{} schedulers~\cite{Tindell}.
\end{proof}

\subsection{Pessimism}\label{sect:pessimism}
As we stated before, this technique introduces some pessimism. Here is a simple example explaining why. Let us consider a system with two tasks, where $\task_1$ is described in Figure~\ref{fig:requestfunctions}(a)--\ref{fig:requestfunctions}(d), and $\task_2$ contains only one configuration with $C_2^1=1$ and $T_2^1=D_2^1=3$. A way of proving that this system is schedulable is to consider the \RBF{} of all possible scenarios for $\task_1$, and to see if $\task_2$ can process its single unit of work within the three first units of time. For every scenario $\sigma$, it boils down to find a $t\le 3$ such that $\RBF_\sigma(t) + C_2^1 \le t$. This is obviously the case for Figure~\ref{fig:requestfunctions}(a) ($t=2$), (b) ($t=3$) and (c) ($t=2$). We let the reader check the other cases. However, we could not find such a $t$ with $\MRBF$, which implies that our test could not conclude the schedulability of the system.

As a future work, we would like to quantify this pessimism, i.e., study how much exactness we loose by considering \MRBF{} instead of all the \RBF's separately. We hope to be able to bound this pessimism by a constant $C$: if our test attests the non-schedulability of a task system $\task$ on a unit speed CPU, then it is for sure not schedulable on a CPU of speed $C$. 

\subsection{Algorithm}

Computing $\MRBF_i(t)$ can be done in a very similar way as in the context of dynamic priorities \cite{digraph}. We present the algorithm (Algorithm~\ref{Alg:MRBF}) here for the sake of completeness. 


The basic idea is to store in a queue ($\mathcal{Q}$) points where a (dense) scenario can pass (i.e., one step stops at this point. In Figure\ref{fig:requestfunctions}, this corresponds to the rightmost point of a vertical segment). We first add the point $(0,0)$ (line~\ref{alg:Q00}), and then all the points that can be reached from $(0,0)$, i.e., $(T^1, C^1), (T^2, C^2), \dots, (T^N, C^N)$ (lines \ref{alg:forbeg} through \ref{alg:forend}). Of course, we only add a point if it is not already in the queue, as observed in \cite{digraph} (line~\ref{alg:lif}). 

Notice that we only need to get $\MRBF$ until the maximal deadline of tasks with a lower index, represented by $\Delta$ (line~\ref{alg:Delta}).

\begin{algorithm}[ht]
\newcommand{\Q}{\mathcal{Q}}
\newcommand{\adds}{\stackrel{+}{\gets}}
\newcommand{\pop}{\stackrel{\text{pop}}{\gets}}
\caption{Computing $\MRBF_i$}
\label{Alg:MRBF}
\KwData{$[C^1, C^2, \ldots, C^N]$, $[T^1, T^2, \ldots,T^N]$, $\Delta$}

$\MRBF(t) = 0, \forall t\ge 0$\;
$\Q \adds (0,0)$\;				\label{alg:Q00}

\While{$\Q$ not empty}{
	$(t, v) \pop \Q$ ; \tcp{$\Q$ is ordered on $t$}
	\If{$t < \Delta$}{ \label{alg:Delta}
		\ForEach{$C^k, T^k$}{	\label{alg:forbeg}
			$t' \gets t+T^k, v' \gets v+C^k$\;
			$ \forall x \ge t: \MRBF(x) \gets \max\{\MRBF(x), v'\}$\;
			
			\If{$(t', v') \notin \Q$ }{ \label{alg:lif}
				$Q \adds (t', v')$\;  \label{alg:forend}
			
			}
		}
	}
}
\Return $\MRBF$\;
\end{algorithm}

\subsection{Bounding \MRBF}
\label{sec:bounding}
In order to efficiently compute \MRBF, as well as to give some approximation, it is useful to bound this function. Before giving some bound, we will first define some value. The maximal slope is defined as
$U_i^{\max}  \equals \max_k\left\{ \frac{C_i^k}{T_i^k} \right\}$. 
Notice that $U_i^{\max}$ could correspond to several configurations of $\tau_i$. The maximal amount of work is given by $C_i^{\max}  \equals \max_k\left\{ C_i^k \right\}$.

We also need to define $C_i^{U_{\max}}$, the computation requirement corresponding to the maximal slope (or the maximum of them if several configurations correspond to the maximal slope):
$C_i^{U_{\max}} \equals \max_k\left\{ C_i^k : \frac{C_i^k}{T_i^k} = U_i^{\max}\right\}$.

%

\begin{property}\label{prop:MRBF1}
$\forall t: \MRBF_i(t) \le \UB_i(t)\equals U_i^{\max} \times t + C_i^{\max}$.
\end{property}

\begin{proof}
In order to prove that $\MRBF_i$ is always below the line $\UB_i$, we will show that no scenario can ever have any point above this line. First, we remind that a scenario is a set of connected steps, each of them composed of one vertical jump of some value $C_i^k$, followed by a horizontal segment of the corresponding $T_i^k$. One step represents the contribution of one job of $\tau_i$ in its configuration $k$.

The proof is done by contradiction. Let us assume that a scenario contains a point $(t, v)$ above the line, i.e., such that $\UB_i(t) < v$. Without loss of generality, we can assume that this point is the left-most point of the segment it belongs to. By construction, this point is in the middle of the contribution of some configuration $k$ of $\tau_i$, and this contribution starts at $(t, v-C_i^k)$. Note that as $C_i^k \le C_i^{max}$, this point is above the line $f(t) = U_i^{\max} \times t$. We should now find a sequence $k_1, k_2, \dots, k_n$ such that 
$\sum_j C_i^{k_j} = v-C_i^k$ and $\sum_j T_i^{k_j} = t$. This is impossible, as it would mean that, in average, the contributions of the jobs would have a slope larger than $U_i^{\max}$, which is a contradiction.
\end{proof}

\begin{property}\label{prop:MRBF2}
$\forall t: \MRBF_i(t) \ge \LB_i(t) \equals U_i^{\max} \times t$.
\end{property}

\begin{proof}
In order to prove this property, we simply build a scenario $\mathcal{S}$ which respects the condition $\forall t: \mathcal{S}(t) \ge \LB_i(t)$. As by definition, $\MRBF_i(t) \ge \mathcal{S}(t)$ for any scenario $\mathcal{S}$ of $\tau_i$, this will make the proof.
The scenario $\mathcal{S}$ consists in choosing infinitely many times the configuration with the highest slope, i.e., the one corresponding to $U_i^{\max}$. Trivially, this scenario ``touches'' $\LB_i$ on the right side of each horizontal segment. 
\end{proof}

With a similar argument, we could refine our lower bound:
\begin{equation}
\label{eq:LB}
\LB_i(t) \equals U_i^{\max} \times t + (C_i^{\max} - C_i^{U_{\max}}).
\end{equation}
We let the formal proof for a longer version of this paper.


\begin{property}\label{prop:belowLB}
No scenario going through $(t,v)$ such that $\LB_i(t) - C^{\max} > v $ can ever coincide with $\MRBF_i$ after $t$.
\end{property}

\begin{proof}
From Prop.~\ref{prop:MRBF1}, a point below $\LB_i - C^{\max}$ will never be part of a scenario reaching $\LB_i$, for the same reason that no scenario of $\tau_i$ starting at $(0,0)$ can ever reach $\UB_i$. As $\MRBF_i$ is always above $\LB_i$ (Prop.~\ref{prop:MRBF2}), a point below $\LB_i - C^{\max}$ will never be part of a scenario contributing in $\MRBF_i$.
\end{proof}

From this last property, we can strongly improve Algorithm~\ref{Alg:MRBF}. As soon as a point is too far from $\MRBF$, it can be ignored. Line \ref{alg:lif} of Algorithm~\ref{Alg:MRBF} can be re-written as:
\begin{equation}\label{eq:obs2}
\textbf{if } t' \times U^{\max} < v' + C^{\max} \textbf{ and } (t', v') \notin \mathcal{Q}
\end{equation} 
With the same argument as in \cite{digraph}, we may prove that Algorithm~\ref{Alg:MRBF} is polynomially bounded in $\Delta$ and $n_i$.

\section{Open Questions and Conclusion}
In this work, we provide a sufficient schedulability test for the scheduling of \ncgmf{} tasks using \ftp{} schedulers.

There is still a lot of work to be done in that direction. Here are a few examples of our plans for the very near future.

\begin{itemize}
\item We observed that starting from a value we still need to quantify, $\MRBF_i$  follows a periodic pattern of size $C_i^{U_{\max}} \times T_i^{U_{\max}}$. We hope to find a fast way to compute this pattern and when it starts. This might allow strongly improving the complexity of computing $\MRBF_i(t)$.

\item As already stated, our method introduces some pessimism that we would like to quantify, and hopefully bound. 
\item We would like to adapt the technique presented in~\cite{RBF} in order to define a polynomial approximation scheme for our schedulability test, combining $\MRBF_i$ for $k$ steps with $\UB_i$ for higher values. However, the complexity of computing $k$ steps of \MRBF{} is much more complex than in~\cite{RBF} where computing $k$ steps is in $\mathcal{O}(k)$. We will need to bound the complexity of getting $k$ steps of \MRBF ; the previous item will be a first step before we get this result.
\item We also plan to adapt our technique to the more general model of Digraph tasks \cite{digraph}. We believe that it should not be very difficult, but we will need to adapt the proof of the critical instant and the algorithms.
\end{itemize}

\section*{Acknowledgment}
We would like to warmly thank Sanjoy Baruah, from the University of North Carolina (Chapel Hill), for posing the problem and for the very interesting and useful discussions we had together.

\bibliographystyle{acm}
\bibliography{RTSS2011-WiP.bib}

\end{document}